\documentclass{lmcs}
\pdfoutput=1

\usepackage{lastpage}
\lmcsdoi{15}{1}{29}
\lmcsheading{}{\pageref{LastPage}}{}{}%
{Oct.~10,~2017}{Mar.~22,~2019}{}

\usepackage{hyperref}
\usepackage{amssymb,xcolor,verbatim,soul}
\usepackage{lscape,longtable,multicol,multirow,bm}
\usepackage{scrextend}
\usepackage{array}

\newcolumntype{M}[1]{>{\centering\arraybackslash}m{#1}}
\newcolumntype{N}{@{}m{0pt}@{}}

\theoremstyle{plain}
\newtheorem{example}[thm]{Examples}

\newcommand\twoheaduparrow{\mathord{\rotatebox[origin=c]{90}{$\twoheadrightarrow$}}}
\newcommand\twoheaddownarrow{\mathord{\rotatebox[origin=c]{90}{$\twoheadleftarrow$}}}

\def\e#1{\emph{#1}}

\def\t#1{\text{#1}}
\def\mc#1{\mathcal{#1}}
\def\sub{\subseteq~\!\!\!}

\def\dda{\mathord{\twoheaddownarrow}_{\mathcal{I}}}

\def\lli{\ll_{\mathcal{I}}}
\def\dua{\mathord{\twoheaduparrow}_{\mathcal{I}}}

\newcommand{\siconv}{\xrightarrow{\mc{I}}}
\newcommand{\brackets}[1]{\left(#1\right)}
\newcommand{\lra}{\longrightarrow}
\newcommand{\ub}{\operatorname{ub}}

\newcommand{\intr}{\operatorname{int}}
\newcommand{\cl}{\operatorname{cl}}

\newcommand{\SI}{\operatorname{SI}}

\newcommand{\bigsup}{\bigvee}
\newcommand{\biginf}{\bigwedge}
\newcommand{\elb}[1]{\operatorname{elb}(#1)}

\newcommand{\da}{\mathord{\downarrow}}
\newcommand{\ua}{\mathord{\uparrow}}
\begin{document}

\title{A Topological Scott Convergence Theorem}

%\author{%
%Hadrian Andradi\affmark[1,2], Weng Kin Ho\affmark[1]\\
%\affaddr{\affmark[1]National Institute of Education, Nanyang Technological University, 1 Nanyang Walk, Singapore 637616}\\
%\affaddr{\affmark[2]Department of Mathematics, Faculty of Mathematics and Natural Sciences, Universitas GadjahMada, Indonesia 55281}\\
%\email{\{A,B,C,D,E\}@university.edu}\\
%}

\author[H. Andradi]{Hadrian Andradi\rsuper{{a,b}}}	%required
\author[W. K. Ho]{Weng Kin Ho\rsuper{a}}	%optional
\address{\lsuper{a} National Institute of Education, Nanyang Technological University, Singapore}	%required
\email{hadrian.andradi@gmail.com}  %optional
\address{\lsuper{b} Department of Mathematics, Universitas Gadjah Mada, Indonesia}	%optional
\email{wengkin.ho@nie.edu.sg}  %optional
\thanks{The first author is supported by Nanyang Technological University Research Scholarship (RSS);
the second author is supported by `Computing with Infinite Data' (731143-CID)}	%optional
%\thanks{thanks 2, optional.}	%optional

%% required for running head on odd and even pages, use suitable
%% abbreviations in case of long titles and many authors:

%% mandatory lists of keywords and classifications:
%\titlecomment{OPTIONAL comment concerning the title, \eg, if a variant
%or an extended abstract of the paper has appeared elsewehere}
%%%%%%%%%%%%%%%%%%%%%%%%%%%%%%%%%%%%%%%%%%%%%%%%%%%%%%%%%%%%%%%%%%%%%%%%%%%

%% the abstract has to PRECEED the command \maketitle:
%% be sure not to issue the \maketitle command twice!
%\linenumbers
\begin{abstract}
 \noindent Recently, J. D. Lawson encouraged the domain theory community to consider the scientific program of developing domain theory in the wider context of $T_0$ spaces instead of restricting to posets.  In this paper, we respond to this calling with an attempt to formulate a topological version of the Scott Convergence Theorem, i.e., an order-theoretic characterisation of those posets for which the Scott-convergence $\mc{S}$ is topological.  To do this, we make use of the $\mathcal{ID}$ replacement principle  to create topological analogues of well-known domain-theoretic concepts, e.g., $\mathcal{I}$-continuous spaces correspond to continuous posets, as $\mc{I}$-convergence corresponds to $\mc{S}$-convergence.
In this paper, we consider two novel topological concepts, namely, the $\mathcal{I}$-stable spaces and the $\mathcal{DI}$ spaces, and as a result we obtain some necessary (respectively, sufficient) conditions under which the convergence structure $\mathcal{I}$ is topological.
\end{abstract}

\maketitle

\section{Introduction}
\label{sec: intro}
Domain theory can be said to be a theory of approximation on partially ordered sets.
There are two sides of the same domain-theoretic coin: the order-theoretic one and the topological one.  On the order-theoretic side, the facility to approximate is built in the ordered structures via \emph{approximation relations}, and here domain is the generic term that includes all ordered structures that satisfy some approximation axioms.  On the topological side, approximation can be handled by \emph{net convergence structures} in $T_0$ spaces.  Amongst many beautiful results in domain theory, one result stands out in that it epitomises this deep connection between order and topology for domains:
\begin{thm}[Scott convergence theorem~{\cite[Theorem II-1.9]{gierzetal03}}]
\label{thm: scott conv thm}
For a directed complete poset $P$, the following are equivalent:
\begin{enumerate}
\item $\mathcal{S}$-convergence is the topological convergence for the Scott topology, i.e., for all $x \in P$ and all nets $(x_i)_{i \in I}$ in $P$,
    \[
    ((x_i)_{i \in I},x) \in \mathcal{S} \iff (x_i)_{i \in I} \text{ converges to } x \text{ with respect to } \sigma(P);
    \]

\item $P$ is a domain (i.e., a continuous dcpo).
\end{enumerate}
\end{thm}
An instance of $((x_i)_{i \in I},x) \in \mathcal{S}$ is denoted by
$x \equiv_{\mathcal{S}} \lim x_i$
(or by $(x_i)_{i \in I} \xrightarrow{\mathcal{S}} x$).
Here the $\mathcal{S}$-convergence is the relation between the class $\Psi P$ of all nets in $P$ and the dcpo $P$ defined by $x \equiv_\mathcal{S} \lim x_i$ if and only if there exists a directed set $D$ of eventual lower bounds of $(x_i)_{i \in I}$ such that $\bigsup D \geq x$.  Note that D. S. Scott's original definition of $\mathcal{S}$-convergence which was formulated for complete lattices used instead the defining condition
\[
x \equiv_{\mathcal{S}} \lim x_i \iff
\underline{\lim} x_i := \bigsup_j \biginf_{i \geq j} x_i \geq x,
\]
and it is not hard to verify that the $\mathcal{S}$-convergence defined for dcpo's when restricted to complete lattices is equivalent to Scott's original definition.
%Because Scott's original formulation makes use of lim-inf's, the $\mathcal{S}$-convergence has been termed as \emph{lim-inf convergence} since its first appearance in \emph{A Compendium of Continuous Lattices}
%(see \cite[p. 104]{gierzetal80}) and subsequently its expanded form
%(see \cite[p. 133]{gierzetal03}).  Incidentally, in both books the term ``lim-inf convergence" is overloaded to include another convergence salient to the Lawson topology (see \cite[p. 158]{gierzetal80} and \cite[p. 231]{gierzetal03}).  In view of this possible confusion,
We shall use the term \emph{Scott convergence} to refer to $\mathcal{S}$-convergence -- crediting this terminology to M. Ern\'{e} who first used it in~\cite{erne81} to mean generalisations of the $\mathcal{S}$-convergence extended to posets (expressed in terms of filters) that apply respectively to the collection of the Frink ideals, the ideals and those ideals whose join exists (denoted by $I_m(P)$, where $m =1,2,3$).  Another generalisation of the Scott convergence theorem for posets was later given independently by B. Zhao and D. Zhao using net convergence \cite{zhaozhao05}.  It can be shown that \cite[Theorem 1]{zhaozhao05} can be derived as a special case of~\cite[Corollary 2.14]{erne81} when $m = 3$.

Our paper aims to formulate and prove a topological variant of Theorem~\ref{thm: scott conv thm} \`{a} la Lawson, which we now explain.  In an invited presentation\footnote{This talk which bore an `extra-terrestrial' title of ``\emph{Close Encounters of the Third Kind: Domain Theory Meets $T_0$ Space Meets Topology}" took place on 27 October 2013 at ISDT'13 in Changsha, China.} at the 6th International Symposium in Domain Theory, J. D. Lawson gave evidence from recent development in domain theory to highlight the intimate relationship between domains and $T_0$ spaces.  In particular, he pointed out that ``several results in domain theory can be lifted from the context of posets to $T_0$ spaces".  We call this research enterprise \emph{developing domain-theory \`{a} la Lawson}.  Forerunners in this line of research include: (1) the topological technique of dcpo-completion of posets~\cite{zhaofan07} can be upgraded to yield a d-completion of $T_0$ spaces (i.e., a certain completion of $T_0$ spaces to yield d-spaces)~\cite{keimellawson09}, and (2) an important order-theoretic result known as Rudin's lemma~\cite{gierzetal83}, which is central to the theory of quasicontinuos domains, also has a topological parallel~\cite{heckmannkeimel13}.

In this paper, we respond to Lawson's call to develop the core of domain theory directly in $T_0$ spaces by establishing a topological parallel of the Scott Convergence Theorem.  Our style of presentation is closer to that of B. Zhao and D. Zhao \cite[Theorem 2.1]{zhaozhao05} than to Ern\'{e}'s
\cite[Corollary 2.14]{erne81} as we consider a net convergence class.
In our attempt to obtain such a topological variant of this result, we adopt the recent approach in~\cite{zhaoho15} and~\cite{heckmannkeimel13} by replacing directed subsets with irreducible subsets.  The motivation for this approach is based on the observation that the Alexandroff irreducible subsets of a poset are precisely the directed subsets and so we call this approach \e{the $\mathcal{ID}$ replacement principle}.
In~\cite{zhaoho15}, Zhao and Ho established that irreducible subsets play a crucial role in the theory of $T_0$ spaces in very much the same way as directed sets do in domain theory.  In particular, since sobriety of topological spaces can be defined in terms of irreducible sets, the approach taken in~\cite{heckmannkeimel13} and~\cite{zhaoho15} allows one to study various types of sobriety for general $T_0$ spaces, which are applicable to the specific case of Scott spaces of posets.

Relying on the $\mathcal{ID}$ replacement principle, some topological analogues of the usual domain-theoretic notions have already been manufactured in~\cite{zhaoho15}:
\begin{enumerate}
\item the irreducibly derived topology on a given $T_0$ space;
\item the relation $\ll_{\SI}$ (denoted by $\ll_{\mathcal{I}}$ in this paper) on $T_0$ spaces; and
\item $\SI$-continuous $T_0$ spaces.
\end{enumerate}

Here we point out that the collection $\mathcal{I}$ of all irreducible subsets of a given space $(X,\tau)$ is a particular instance of a subset system $\mathcal{M}$ of the specialisation poset of $(X,\tau)$ in the sense of~\cite{erne81a,bandelterne84,yao11,zhouzhao07} -- a fact that was neither mentioned nor exploited in \cite{zhaoho15}.

In order to establish a topological parallel of the Scott Convergence Theorem, we consider the notion of $\mathcal{I}$-convergence and $\mathcal{I}$-continuity. More precisely, we consider a net convergence in any $T_0$ space $X$, called $\mathcal{I}$-convergence, which is defined by replacing directed sets with irreducible sets in the definition of $\mathcal{S}$-convergence, i.e., for any $x \in X$ and any net $(x_i)_{i \in I}$ in $X$, $\brackets{x_i}_{i \in I} \siconv x$ if and only if there exists an irreducible set $E$ of eventual lower bounds of $(x_i)_{i \in I}$ such that $\bigsup E \geq x$ -- it is a special instance of the known $\mathcal{M}$-Scott convergence (see, e.g., \cite{yao11}) or lim-inf$_{\mathcal{M}}$-convergence (see, e.g., \cite{zhouzhao07}). Similarly, the definition of $\mathcal{I}$-continuous spaces is formulated by applying $\mathcal{ID}$ replacement principle to the definition of continuous posets. We then restrict our attention to certain subclasses of $T_0$ spaces, i.e., the $\mathcal{I}$-stable spaces and $\mathcal{DI}$ spaces -- which crucially include all Alexandroff topologies on posets.
With the additional ammunition of $\mathcal{I}$-continuity and the new classes of $T_0$ spaces at our disposal, our topological variant of the Scott Convergence Theorem consists of two independent parts:
\begin{enumerate}
\item If a space is $\mathcal{I}$-stable and $\mathcal{I}$-continuous, then the $\mathcal{I}$-convergence in it is topological.
\item A $\mathcal{DI}$ space with topological $\mathcal{I}$-convergence is $\mathcal{I}$-continuous.  In such a case, it is even continuous with respect to its specialisation order.
\end{enumerate}
By specialising our Topological Scott Convergence Theorem to the Alexandroff topologies on posets, we can recover the original Scott Convergence Theorem.

Here is how we organise this paper.  We gather the preliminaries that are needed for the ensuing development in Section~\ref{sec: prelim}.  The needed concepts include some useful facts about irreducibly derived topologies and sup-sober spaces from~\cite{zhaoho15} as well as basic terminologies concerning net convergence.  In Section~\ref{sec: in the light of Ztheory}, we look at  the $\mathcal{ID}$ replacement principle within a broader framework of $\mathcal{M}$-subset system. We then focus on the collection $\mathcal{I}$ of all irreducible sets of a given $T_0$ space in Section~\ref{sec: balanced and spaces}. We introduce a new type of $T_0$ spaces called the $\mathcal{I}$-stable spaces.  The upshot is that for $\mathcal{I}$-stable spaces, $\mathcal{I}$-continuity is indeed sufficient to guarantee that $\mathcal{I}$-convergence is topological.
Then, we consider the other new class of $T_0$ spaces, called the $\mathcal{DI}$ spaces, and arrive to the second part of our topological variant of Scott Convergence Theorem.

Some standard references for domain theory include ~\cite{abramskyjung94,gierzetal03,goubaultlarrecq2013}.
Readers looking for more comprehensive study concerning $\mathcal{M}$-subset system should consult~\cite{bandelterne84, erne81a, yao11, zhouzhao07}.

\section{Preliminaries}
\label{sec: prelim}
In this section, we gather at one place all the preliminaries needed for the present theoretical development, most of which are recalled from~\cite{zhaoho15}, and hence their proofs are omitted.

A nonempty subset $E$ of a topological space $(X,\tau)$ is \emph{irreducible} if for any closed sets $A_1$ and $A_2$, whenever $E \subseteq A_1 \cup A_2$, either $E \subseteq A_1$ or $E \subseteq A_2$ holds.  The collection of all irreducible subsets of $X$ is denoted by $\mathcal{I}$.  Note that a nonempty set $E$ is irreducible if and only if for any finite number $n$ of open sets $U_i$ with $E \cap U_i \neq \varnothing$ ($i < n$), one has $E \cap \bigcap_{i < n} U_i \neq \varnothing$.

Every $T_0$ space $(X,\tau)$ can be viewed as a partially ordered set via its \emph{specialisation order}, denoted by $\leq_\tau$, where $x \leq_{\tau} y$ if $x \in \cl_{\tau}(y)$.  We call the poset $(X,\leq_\tau)$ the \emph{specialisation poset of} $(X,\tau)$.  Henceforth, all order-theoretical statements regarding a $T_0$ space refer to its specialisation order.  For any subset $A$ of a $T_0$ space $(X,\tau)$, the supremum of $A$, if it exists, denoted by $\bigsup_{\tau} A$ or simply $\bigsup A$, is the least upper bound of $A$ with respect to the specialisation order $\leq_\tau$, or simply $\leq$, of the space. We denote the set of all irreducible subsets of a $T_0$ space $X$ whose suprema exist by $\mathcal{I}_{\vee}$.

Scott topology is a prominent topology considered in domain theory.
Since a Scott open subset of a poset $P$ is defined to be an upper set that is in addition inaccessible by directed suprema, the Scott topology $\sigma(P)$ is coarser than the Alexandroff topology $\alpha(P)$ on the poset.  By applying the $\mathcal{ID}$ replacement principle to the definition of a Scott open set, one defines on any $T_0$ space a coarser topology called the \emph{irreducibly derived topology} that mimics the Scott topology on a poset.  More precisely, let $(X,\tau)$ be a $T_0$ space and $U \subseteq X$, define $U \in \tau_{\SI}$ if
(1) $U \in \tau$, and (2) for every $E \in \mathcal{I}_{\vee}$, $\bigsup_\tau E \in U$ implies $E \cap U \neq \varnothing$. It can be readily verified that $\SI(X,\tau) := \brackets{X,\tau_{\SI}}$ is a topological space whose topology is coarser than $\tau$. An open set in $\SI(X,\tau)$ is called \emph{$\SI$-open} and the interior of a subset $A$ of $X$ with respect to $\tau_{\SI}$ is denoted by $\intr_{\tau_{\SI}}(A)$.
Because the Scott-like topology $\tau_{\SI}$ is derived from a topology $\tau$ on the same set $X$, we sometimes refer to $\tau_{\SI}$ as the \emph{Scott derivative} of $\tau$.  Of course, $\SI(\textup{A} P) = \Sigma P$, where $\textup{A} P = (P,\alpha(P))$ and $\Sigma P = (P,\sigma(P))$.

Just as the Scott topology of a given poset does not generally coincide with its Alexandroff topology, so is it with the Scott derivative of a $T_0$ space and its original topology.
Regarding spaces that enjoy the aforementioned coincidence of topologies, we have something to say about them with regards to sobriety. Recall that a topological space $X$ is \emph{sober} if every irreducible closed set is the closure of a unique singleton.  It is well known that the Scott topology on any continuous domain is sober.  A weaker form of sobriety is \emph{bounded-sobriety} which only requires that every irreducible closed set which is bounded above is the closure of a unique singleton. Bounded-sober spaces have been studied in~\cite{mislove99} and~\cite{zhaofan07}. An even weaker form of sobriety is sup-sobriety.  A $T_0$ space $(X,\tau)$ is \emph{sup-sober} if every closed set $F \in \mathcal{I}_{\vee}$ is the closure of a unique singleton; in this case, $F$ is exactly $\cl\left(\{\bigvee F\}\right)$.  Sup-sober spaces are a commonly encountered type of $T_0$ spaces.  Every $T_1$ space is sup-sober.  Every poset $P$ is sup-sober with respect to its upper topology \cite[Corollary 4.9]{zhaoho15}.  The Scott topology on a continuous poset is always sup-sober~\cite[Theorem 7.9]{zhaoho15}, though not all sup-sober spaces are continuous as witnessed by the Johnstone space~\cite{johnstone81}.
Sup-sober spaces have the following pleasant characterisation:
\begin{thmC}[{\cite[Theorem 4.5]{zhaoho15}}] \label{th:SI-infty}
\label{thm: char of sup-sober space}
A $T_0$ space is sup-sober if and only if it is equal to its Scott derivative.
\end{thmC}

In a topological space, approximation can be described by means of net convergence.
Let $X$ be a set.  A \emph{net} $n := (x_i)_{i \in I}$ in $X$ is a mapping from a directed pre-ordered set $(I,\leq)$ to $X$.
Real number sequences, for instance, are nets in the Euclidean space $\mathbb{R}$.  Thus, nets can be viewed as generalised sequences.  We denote the class of all nets in $X$ by $\Psi X$.

Every $x \in X$ generates for each directed set $I$ a \emph{constant net} $\brackets{x}_{i \in I}$ given by $x_i = x$ for all $i \in I$.  Parallel to the notion of subsequence is that of subnet.  A net $n' := (y_j)_{j\in J}$ is a \emph{subnet} of $n := (x_i)_{i\in I}$ if (i) there exists a function $g: J \lra I$ such that $y_j = x_{g(j)}$ for all $j \in J$ and (ii) for each $i \in I$ there exists a $j' \in J$ such that $g(j)\geq i$ whenever $j \geq j'$.

A \emph{net convergence} in a set $X$ is a relation $\mathcal{C}$ between $\Psi X$ and $X$.  We write $n \xrightarrow{\mathcal{C}} x$ if $(n,x)$ belongs to $\mathcal{C}$, in which case we say that the net $n$ \emph{$\mathcal{C}$-converges} to $x$.  A net convergence $\mathcal{C}$ in $X$ is said to satisfy the Constant-net condition (Constants) (and respectively, the Subnet condition (Subnets)) if:
\begin{enumerate}
\item[] \hspace{-12mm}(Constants)  For any $x \in X$, it holds that $\brackets{\brackets{x}_{i\in I},x}\in \mc{C}$.
\item[] \hspace{-12mm}(Subnets)
If $\brackets{n,x}\in\mc{C}$ and $n'$ is a subnet of $n$, then $\brackets{n',x}\in\mc{C}$.
\end{enumerate}

Every net convergence $\mathcal{C}$ in $X$ induces a topology, where the opens are those $U \subseteq X$ satisfying the following condition:
whenever a net $\brackets{x_i}_{i \in I} \xrightarrow{\mc{C}} x$ and $x \in U$, then $x_i \in U$ eventually.  In the opposite direction, every space $(X,\tau)$ induces a net convergence $\mathcal{C}_\tau$ defined by
$(x_i)_{i \in I} \xrightarrow{\mathcal{C}_\tau} x$ if
$\forall U \in \tau.~(x \in U \implies x_i \in U$  eventually).  Here, a property of a net $(x_i)_{i \in I}$ \emph{holds eventually} if there exists an $i_0 \in I$ such that for all $i \geq i_0$, the property holds for $x_i$.

Given a set $X$ and a topology $\tau$ on $X$, when $n \xrightarrow{\mathcal{C}_\tau} x$, we say that \emph{$n$ converges to $x$ with respect to the topology $\tau$}. We shall sometimes write $n\xrightarrow{\tau}x$ to indicate $n\xrightarrow{\mathcal{C}_{\tau}}x$.  A net convergence $\mathcal{C}$ in a set $X$ is said to be \emph{topological} if there is a topology $\tau$ on $X$ that induces it, i.e., $\mathcal{C} = \mathcal{C}_{\tau}$. In fact, for a topological convergence class, the topology inducing it is unique owing to the following proposition:

\begin{propC}[{\cite[page 76]{kelley55}}] \label{prop:top_containment_char}
Let $\tau$ and $\sigma$ be topologies on a set $X$.
Then $\tau \sub \sigma$ if and only if $\mc{C}_{\sigma} \sub \mc{C}_{\tau}$.
\end{propC}

In what follows, all topological spaces are assumed to be $T_0$ spaces.

\section{Generalised continuity and Scott convergence}
\label{sec: in the light of Ztheory}

Given a poset $P = (X,\leq)$, let $\mathcal{M}$ be a collection of subsets of the set $X$ containing all singletons, and denote by
\begin{enumerate}
\item[] $\ub(A)$ the set of all upper bounds of $A \subseteq P$,
\item[] $\elb{n}$ the set of all eventual lower bounds of a net $n$ in $P$,
\item[] $\mathcal{M}_{\vee}$ the set of all members of $\mathcal{M}$ that have a supremum,
\item[] $\mathcal{M}^{\wedge}$ the set of all lower sets generated by members of $\mathcal{M}$.
\end{enumerate}
Further, define the $\mathcal{M}$-below relation (cf. \cite{baranga96, bandelterne83, erne99, venugopalan86}) by
$$x \ll_{\mathcal{M}} y \iff \forall M \in \mathcal{M}_{\vee} (y \leq \bigvee M \implies x \in \da M)$$
and put (cf. \cite{baranga97, menon96})
$$\sigma_{\mathcal{M}}= \{U = \ua U \mid \forall M \in \mathcal{M}_{\vee} (\bigvee M \in U \implies U \cap M \neq \varnothing) \}.$$
This set is closed under arbitrary unions, hence the notion of interior is still relevant, but not always a topology. But from the definition of irreducible sets, given a $T_0$ space $(X,\tau)$, it follows that $\sigma_{\mathcal{M}}\cap \tau$ is a topology whenever $\mathcal{M}$ is contained in the collection of all irreducible sets in $(X,\tau)$.

For $Y \subseteq P$, put
%$$
%\intr_{\sigma_{\mathcal{M}}}(Y) = \bigcup \{U \in %\sigma_{\mathcal{M}} \mid U \subseteq Y\}
%$$
%and
$$
\twoheaddownarrow_{\mathcal{M}} Y :=
\{x \in P \mid \exists y \in Y.~(x \ll_\mathcal{M} y)\}
~ \text{ and }~
\twoheaduparrow_\mathcal{M} Y :=
\{x \in P \mid \exists y \in Y.~(y \ll_\mathcal{M} x)\}.
$$
For any element $y$ of $P$, the set $\twoheaddownarrow_{\mathcal{M}} y$ (resp., $\twoheaduparrow_{\mathcal{M}} y$) refers to $\twoheaddownarrow_{\mathcal{M}} \{y\}$ (resp., $\twoheaduparrow_{\mathcal{M}} \{y\}$).

\begin{rem}
In some works, instead of $\mathcal{M}$, the letter $\mathcal{Z}$ is used when ones study a subset selection, i.e., a certain assignment assigning each poset to a certain collection of posets (cf.~\cite{bandelterne83, baranga96, erne99, venugopalan86, wrightwagnerthatcer78, zhao92}). Although we know that most of definitions and results presented in this section can be performed in $\mathcal{Z}$-subset selection language, we choose to use the letter $\mathcal{M}$ as in \cite{erne81a, bandelterne84, yao11, zhouzhao07} since our main focus is on a particular collection of subsets of the underlying set of a given poset.
\end{rem}

We now consider the generalised version of $\mathcal{S}$-convergence as follows (c.f. \cite{yao11, zhouzhao07}):
\begin{defi}[$\mathcal{M}$-convergence]
A net $n$ in a poset $P$ \emph{$\mathcal{M}$-converges} to a point $x$, denoted by $n \xrightarrow{\mathcal{M}} x$, if there exists an $M \in \mathcal{M}_{\vee}$ with $M \subseteq \elb{n}$ and $x \leq \bigsup M$. We write $\tau_{\mathcal{M}}$ to denote the topology induced by $\mathcal{M}$-convergence.
\end{defi}

It follows easily that:
\begin{lem} \label{lem: I satisfies Con and Sub}
The $\mathcal{M}$-convergence in any poset satisfies the conditions (Constants) and (Subnets).
\end{lem}

\begin{lem}\label{lem: tau_m char}
For any  poset $P$, $U \in \tau_{\mathcal{M}}$ if and only if the following conditions hold:
\begin{enumerate}
\item[{\rm (1)}] $U = \ua U$,
\item[{\rm (2)}] for every $M \in \mathcal{M}_{\vee}$ with $\bigvee M \in U$, there exists a finite subset $N$ of $M$ such that $\ub(N) \subseteq U$.
\end{enumerate}
\end{lem}
\begin{proof}
Let $V\in \tau_{\mathcal{M}}$.
\begin{enumerate}
\item[(1)] Let $x\in \ua U$. There exists $u \in U$ such that $u \leq x$. Clearly, the constant net generated by $x$ $\mc{M}$-converges to $u$. Hence $x\in U$.
\item[(2)] Let $M\in\mathcal{M}_{\vee}$ such that $\bigvee M\in U$. Define $$I_M = \{(u, N) \mid u \in \ub(N), N\subseteq_{\rm fin} M\}$$ and equip it with the order defined as follows: $(u_1, N_1) \leq (u_2, N_2)$ if and only if $\ub(N_2) \subseteq \ub(N_1)$. Clearly, $I_M$ is directed. We consider the net $n:=(x_i)_{i\in I_M}$, with $x_{(u, N)} = u$. It can be verified that $n \xrightarrow{\mathcal{M}} \bigvee M$. Since $U\in \tau_{\mathcal{M}}$, there exists an $i_0:=(u, N )\in I_M$ such that $j\geq i_0$ implies $x_j \in U$. Then for each $t\in\ub(N)$ it holds that $j:=(t, N)\geq i_0$, we have that $x_j = t \in U$. Therefore $\ub(N)\subseteq U$.
\end{enumerate}
Conversely, let $U$ satisfy (1) and (2), and $n \xrightarrow{\mathcal{M}} x\in U$. There exists an $M\in\mathcal{M}_{\vee}$ such that $x\leq \bigvee M$ and $M\subseteq\elb{n}$. By (1), $\bigvee M \in U$. By (2), there exists a finite subset $N$ of $M$ such that $\ub(N) \subseteq U$. Since $N$ is finite and $N \subseteq \elb{n}$, $n$ is eventually in $\ub(N)$. Therefore $U \in \tau_{\mc{M}}$.
\end{proof}

\begin{lem} \label{lem: topology ind by I conv is finer than the irred der topology}
For every poset $P$, $\sigma_\mathcal{M} \subseteq \tau_{\mathcal{M}}$, i.e., the topology induced by the $\mathcal{M}$-convergence is finer than the topology generated by $\sigma_{\mathcal{M}}$.
\end{lem}
\begin{proof}
It is immediate from Lemma \ref{lem: tau_m char}.
\end{proof}

\begin{rem} \label{rem: strict containment}
Let $P:=\{a,b,c\}$ be equipped with a partial order $\leq$ such that $a$ and $b$ are incomparable and $a,~b \leq c$. Let $\mathcal{N}$ be the set of all antichains in $P$. Since $\{c\} \in \tau_{\mathcal{N}}-\sigma_{\mathcal{N}}$, it follows that $\sigma_{\mathcal{N}} \subsetneq \tau_{\mathcal{N}}$. In this case, the collection $\sigma_\mathcal{N}$ is not a topology since $\ua a, \ua b \in \sigma_{\mathcal{N}}$, but $\ua a \cap \ua b = \{c\} \notin \sigma_{\mathcal{N}}$.
\end{rem}

\begin{prop}
\label{prop: transitivity of lli}
For any $u,~x,~y$ and $z$ in a poset $P$,
\begin{enumerate}
\item[{\rm (1)}] $x \ll_{\mathcal{M}} y$ implies $x \leq y$.

\item[{\rm (2)}] $u \leq x \ll_{\mathcal{M}} y\leq z$ implies $u \ll_{\mathcal{M}} z$.

\item[{\rm (3)}] $\intr_{\sigma_{\mathcal{M}}}(\ua y) \subseteq \twoheaduparrow_{\mathcal{M}} y$.
\end{enumerate}
\end{prop}

\begin{defi}[$\mathcal{M}$-continuous poset (\cite{bandelterne83,baranga97,erne99,venugopalan86})]
\label{defi: Z-continuous poset}
A poset $P$ is \emph{$\mathcal{M}$-continuous} if each $y \in P$ satisfies the conditions
$\twoheaddownarrow_{\mathcal{M}} y \in \mathcal{M}^{\wedge}$ and
$y = \bigsup \twoheaddownarrow_{\mathcal{M}} y$.
If, moreover, $\ll_{\mathcal{M}}$ satisfies the interpolation property (i.e., $x \ll_{\mathcal{M}} z$ always implies that $x \ll_{\mathcal{M}} y \ll_{\mathcal{M}} z$ for some $y$), then $P$ is \emph{strongly $\mathcal{M}$-continuous}.
\end{defi}
\begin{rem}
Our definition of $\mathcal{M}$-continuous poset follows that in~\cite{baranga97} and differs from that in~\cite[p.~53]{erne99}. More precisely, saying a poset $P$ is $\mathcal{M}$-continuous posets in the present definition is equivalent to saying that $P$ is $\mathcal{Z}_{\vee}$-precontinuous poset in~\cite{erne99}, where $\mathcal{Z}_{\vee}P = \mathcal{M}_\vee$.
\end{rem}

The lemmas below follow immediately from the definitions of $\mathcal{M}$-convergence and $\mathcal{M}$-continuity.
\begin{lem} \label{lem:Z-conv_equiv}
Let $P$ be a poset and $n$ a net in $P$. If $n \xrightarrow{\mathcal{M}} y$, then
$\twoheaddownarrow_{\mathcal{M}} y \subseteq \elb{n}$.
The converse holds if $P$ is $\mathcal{M}$-continuous.
\end{lem}

\begin{lem}\label{lem:interior_in-Z-cts}
If $P$ is $\mathcal{M}$-continuous, then $\intr_{\sigma_{\mathcal{M}}}(\ua Y) \subseteq \twoheaduparrow_{\mathcal{M}} Y$ holds for all $Y \subseteq P$.
\end{lem}
\begin{comment}
\begin{proof}
If $x \in \intr_{\sigma_{\mathcal{M}}}(\ua Y)$ then $\intr_{\sigma_{\mathcal{M}}}(\ua Y) \cap \twoheaddownarrow_{\mathcal{M}} x \neq \varnothing$, hence there exists a $z \in \intr_{\sigma_{\mathcal{M}}}(\ua Y) \subseteq \ua Y$ such that $z \ll_{\mathcal{M}} x$. Therefore, $y \ll_{\mathcal{M}} x$ for some $y \in P$.
\end{proof}
\end{comment}

\begin{lem} \label{prop: subset thadai x}
A poset $P$ is $\mathcal{M}$-continuous if and only if for all
$y \in P$ there is an $M \in \mathcal{M}_{\vee}$ with
$M \subseteq \twoheaddownarrow_{\mathcal{M}} y$ and $y \leq \bigsup M$.
\end{lem}

\begin{lem} \label{lem: int implies way-above open}
If the relation $\ll_{\mathcal{M}}$ on a poset $P$ is interpolative then
$\twoheaduparrow_{\mathcal{M}} Y \in \sigma_{\mathcal{M}}$ for all $Y \subseteq P$.
\end{lem}

\begin{lem} \label{lem: prelude to interpolating}
If $P$ is an $\mathcal{M}$-continuous poset, then $y = \bigsup \twoheaddownarrow_{\mathcal{M}} (\twoheaddownarrow_{\mathcal{M}} y)$ for all $y \in P$.  Moreover,
$P$ is strongly continuous if and only if $\twoheaddownarrow_{\mathcal{M}} (\twoheaddownarrow_{\mathcal{M}} y) \in \mathcal{M}^{\wedge}$ and $y = \bigsup \twoheaddownarrow_{\mathcal{M}} (\twoheaddownarrow_{\mathcal{M}} y)$ for all $y \in P$.
\end{lem}

\section{$\mathcal{I}$-convergence and $\mathcal{I}$-continuous $\mathcal{I}$-stable spaces}
\label{sec: balanced and spaces}
From this juncture onwards, we consider the collection $\mathcal{I}$ of all irreducible sets in a given space $(X,\tau)$, which, by considering the specialisation poset of the space, is a special case of $\mathcal{M}$ given in Section \ref{sec: in the light of Ztheory}. Note that to generate $\mathcal{I}$, ones need to start with a $T_0$ space -- not just a general poset. %Our work will be more focused on the collection of all irreducible sets having a supremum, that is, $\mathcal{I}_{\vee}$.

\begin{rem}\label{rem: sigma I is a top}
Given any $T_0$ space $(X,\tau)$, $\tau_{\SI} = \tau \cap \sigma_{\mathcal{I}}$ is a topology on $X$. The topology $\tau_{\SI}$ can be strictly coarser than $\tau_{\mathcal{I}}$ as witnessed by the space $\mathbb{N}$ with the cofinite topology.
\end{rem}

Because continuity of a poset $P$ is precisely the characterising property for the Scott convergence in it to be topological, it is natural to ask whether a similar result holds for $\mathcal{I}$-continuity of spaces and $\mathcal{I}$-convergence.  In this section, we introduce some conditions under which the $\mathcal{I}$-convergence is topological.

\begin{defi}[$\mathcal{I}$-continuous space]
\label{defi: Irr-continuous}
A space $X$ is said to be \emph{$\mathcal{I}$-continuous} if its poset of specialisation is $\mathcal{I}$-continuous in the sense of Definition~\ref{defi: Z-continuous poset}.
\end{defi}

The following lemma can be verified easily.

\begin{lem}\label{lem: I-cts char}
A space $X$ is $\mathcal{I}$-continuous if and only if for every $y \in X$, $\twoheaddownarrow_{\mathcal{I}} y \in \mathcal{I}_{\vee}$ and $\bigvee \twoheaddownarrow_{\mathcal{I}} y = y$.
\end{lem}

%In general, the relation $\ll_{\mathcal{I}}$ on $\mathcal{I}$-continuous spaces does not enjoy the interpolation property.
For a continuous poset $P$ and a directed subset $D$ of $P$, the set $\twoheaddownarrow D$ is always directed, which implies $\ll$ is interpolative. This fact is important in obtaining the generalisations of Theorem \ref{thm: scott conv thm} for poset as in \cite{erne81} and \cite{zhaozhao05}. In the context of $T_0$ spaces, it is still unknown to us whether for an $\mathcal{I}$-continuous space $X$ the operation $\dda$  preserves irreducibility. Here we introduce a topological condition called $\mathcal{I}$-stability so that for any $\mathcal{I}$-stable $\mathcal{I}$-continuous space $X$, irreducibility is preserved under $\dda$, and hence $\ll_{\mathcal{I}}$ on the specialisation poset of $X$ is interpolative.

%The fact that the way-below relation $\ll$ on a continuous poset is interpolative is important in obtaining Theorem \ref{thm: scott conv thm}. However, it is still unknown to us whether the relation $\ll_{\mathcal{I}}$ on an $\mathcal{I}$-continuous space always enjoys the interpolation property. Here we introduce a topological condition called $\mathcal{I}$-stability so that $\ll_{\mathcal{I}}$ for any $\mathcal{I}$-stable $\mathcal{I}$-continuous space is interpolative.

\begin{defi}[$\mathcal{I}$-stable space] \label{defi: balanced spaces}
A $T_0$ space $(X,\tau)$ is said to be
\emph{(upwards) $\mathcal{I}$-stable} if
for any $U \in \tau$, $\dua U \in \tau$.
\end{defi}
\begin{example}
\begin{enumerate}
\item[{\rm (1)}] Every poset endowed with the Alexandroff topology is  $\mathcal{I}$-stable.
\item[{\rm (2)}] Scott spaces of continuous posets are $\mathcal{I}$-stable since $\ll = \ll_\mathcal{I}$ (see Examples \ref{ex:nice}(2)  and Remark \ref{rem:DI}).
\item[{\rm (3)}] Any $T_1$ space is $\mathcal{I}$-stable since $U = \dua U$ for every open set $U$.
\item[{\rm (4)}] The Johnstone space $\Sigma \mathbb{J}$ \cite{johnstone81} is not $\mathcal{I}$-stable since $\dua \mathbb{J} = \{(m, \infty) \mid m \in \mathbb{N}\}$ is not in $\sigma(\mathbb{J})$ as a consequence of the fact $\ll = \ll_{\mathcal{I}}$ (see Examples \ref{ex:nice}(4) and Remark \ref{rem:DI}).
\end{enumerate}
\end{example}

For the purpose of establishing an important theorem concerning $\mathcal{I}$-stable spaces, we need to recall the definition of the \emph{lower Vietoris topology}.  Let $(X,\tau)$ be a $T_0$ space.  We endow the space $\mathcal{I}$ of all irreducible subsets of $X$ with the lower Vietoris topology $\nu_{\mathcal{I}}$, i.e., the (subbasic) opens are of the form
\[
\diamond U:= \{E \in \mathcal{I} \mid E \cap U \neq \varnothing\},
\]
where $U \in \tau$.  It is easy to show that every $\nu_{\mathcal{I}}$-open is of the form $\diamond U$ for some $U\in \tau$.  Looking into the literature of $T_0$ spaces, what we have considered above is not entirely new.  For a topological space $(X,\tau)$, define the set $X^S := \{E \in \mathcal{I} \mid X - E \in \tau\}$ and endow it with the lower Vietoris topology $\nu_{\mathcal{I}}$ as above.
The resulting topological space $S(X) := (X^S,\nu_{\mathcal{I}})$ is known as the \emph{sobrification} of $(X,\tau)$~(see also \cite[Exercise V-4.9]{gierzetal03}, \cite[Section 8.2.3]{goubaultlarrecq2013} and \cite[Remark 3.3]{heckmannkeimel13}).

\begin{thm} \label{thm: fund thm on I-stable spaces}
Consider the following statements for an $\mathcal{I}$-continuous space $(X,\tau)$:
\begin{enumerate}

\item[{\rm (1)}] $(X,\tau)$ is sup-sober.

\item[{\rm (2)}] $(X,\tau)$ is $\mathcal{I}$-stable.

\item[{\rm (3)}] $\dda : (X,\tau) \lra (\mathcal{I},\nu_{\mathcal{I}}),~ x \mapsto \dda x$, is continuous.

\item[{\rm (4)}] $U \in \tau$ implies $\dua U = \intr_{\tau_{\SI}} (U)$.

\item[{\rm (5)}] $E \in \mathcal{I}$ implies $\dda E \in \mathcal{I}$.

\item[{\rm (6)}] $\ll_{\mathcal{I}}$ has the interpolation property.

\item[{\rm (7)}] $\dua Y = \intr_{\sigma_{\mathcal{I}}} (\ua Y)$ for all $Y \subseteq X$.

\item[{\rm (8)}] $\dua x \in \sigma_{\mathcal{I}}$ for all $x \in X$.

\item[{\rm (9)}] $\mathcal{I}$-convergence is topological and the topology generated by $\sigma_{\mathcal{I}}$ is exactly $\tau_{\mathcal{I}}$.

\end{enumerate}
In general, the following implications and equivalences hold:
\[
(1) \implies (2) \iff (3) \iff (4) \implies (5) \implies (6) \iff (7) \iff (8) \implies (9).
\]

\end{thm}
\begin{proof}
(1) $\implies$ (2): If $u\in U\in\tau$, then, by $\mathcal{I}$-continuity and Theorem \ref{thm: char of sup-sober space}, $\dda u\cap U\neq\varnothing$, and hence $U\subseteq\dua U\subseteq U$.

(4) $\implies$ (2): Obvious since $\tau_{\SI}\subseteq \tau$.

(2) $\iff$ (3):
Since $(X,\tau)$ is an $\mathcal{I}$-continuous space, by Lemma~\ref{lem: I-cts char}, the mapping $\dda$ in (3) is well-defined.  That $\dda$ is continuous with respect to the topologies $\tau$ and $\nu_{\mathcal{I}}$ if and only if $X$ is $\mathcal{I}$-stable follows from the direct calculations below: for any $U \in \tau$, we have
\[
\dda^{-1}(\diamond U)
= \{x \in X \mid \dda x \cap U \neq \varnothing\}
= \{x \in X \mid \exists y \in U.~x \in \dua y\}= \dua U.
\]

(3) $\implies$ (5): By (3), $\mathcal{E}:=\{\dda e\mid e\in E\}$ is irreducible in $(\mathcal{I},\nu_{\mathcal{I}})$. Note that $\dda E=\bigcup\mathcal{E} $. For any finite number $n$ of open sets $U_i$ with $\dda E\cap U_i\neq\varnothing$ $(i < n)$, it holds that $\mathcal{E}\cap \diamond U_i\neq\varnothing$ $(i <n)$, and hence $\mathcal{E}\cap \bigcap_{i<n} \diamond U_i=\mathcal{E}\cap\diamond\bigcap_{i<n}U_i\neq\varnothing$. So $\dda E\cap \bigcap_{i<n}U_i\neq \varnothing$.
%Obvious since continuous image of irreducible set is irreducible.

(5) $\implies$ (6):
By $\mathcal{I}$-continuity and Lemma~\ref{lem: I-cts char}, for every $z\in X$, $\dda z \in \mathcal{I}$, and thus $\dda (\dda z) \in \mathcal{I}$ by (5).
By Lemma~\ref{lem: prelude to interpolating}, $\bigsup \dda (\dda z) = z$, and so $x \ll_{\mathcal{M}} z$ implies there exists a $y \in \dda z$ such that $x \in \dda y$.

(6) $\implies$ (7): By virtue of Lemmas~\ref{lem:interior_in-Z-cts} and~\ref{lem: int implies way-above open} specialised to $\mathcal{M} = \mathcal{I}$ (by considering the specialisation poset of $(X,\tau)$).

(7) $\implies$ (8): Trivial.

(8) $\implies$ (6): Assume $z \in \dua x$. By $\mathcal{I}$-continuity and Lemma~\ref{lem: I-cts char}, $\dda z \in \mathcal{I}_{\vee}$ and $\bigsup \dda z = z$.  So, by (8), $\dda z \cap \dua x\neq\varnothing$, and hence $\ll_{\mathcal{I}}$ has the interpolation property.

(2)+(7) $\implies$ (4): By (2), (7), and Lemma \ref{lem: topology ind by I conv is finer than the irred der topology}, $\dua U\in\tau\cap \sigma_{\mathcal{I}}=\tau_{\SI}$, and hence $\dua U \subseteq \intr_{\tau_{\SI}}(U)$. By virtue of  (2) and (7) and the fact that $\tau_{\SI}\subseteq \sigma_{\mathcal{I}}$, we have that $\intr_{\tau_{\SI}}(U)\subseteq \intr_{\sigma_{\mathcal{I}}}(U)= \dua U\subseteq \intr_{\tau_{\SI}}(U)$. Therefore, $\dua U = \intr_{\tau_{\SI}}(U)$.

(8) $\implies$ (9): By definition of $\tau_{\mathcal{I}}$, $n \xrightarrow{\mathcal{I}} y$ implies $n \xrightarrow{\tau_\mathcal{I}} y$. Now let $n \xrightarrow{\tau_{\mathcal{I}}} y$. By $\mathcal{I}$-continuity, $\dda y \in \mathcal{I}_{\vee}$ and $y = \bigvee \dda y$. If $x \ll_{\mathcal{I}} y$, then, by (8) and Lemma~\ref{lem: topology ind by I conv is finer than the irred der topology}, $y \in \dua x \in \tau_\mathcal{I}$. Hence $n$ is eventually in $\dua x$. By Proposition \ref{prop: transitivity of lli}(1)  we have that $x \in \elb{n}$.  Therefore $n \xrightarrow{\mathcal{I}} y$. We have that $\mathcal{I}$-convergence is topological. From Lemma~\ref{lem: topology ind by I conv is finer than the irred der topology}, we already have that $\sigma_{\mathcal{I}} \subseteq \tau_{\mathcal{I}}$. Now let $x \in U \in \tau_\mathcal{I}$. By $\mathcal{I}$-continuity and Lemma \ref{lem: tau_m char}, there exist $y_1, y_2,\ldots, y_k \in \twoheaddownarrow_{\mathcal{I}} x$ such that $\ub(\{y_1, y_2,\ldots, y_k\}) \subseteq U$. By virtue of (8), $U_i:= \twoheaduparrow_{\mathcal{I}}y_i \in \sigma_{\mathcal{I}}$ for each $i=1,2,\ldots,n$. Clearly, $x \in V:=\bigcap_{i=1}^k U_i$. If $y \in V$, then by Proposition \ref{prop: transitivity of lli}(1), $y \in \ub(\{y_1,y_2,\ldots, y_k\})$, and hence $V\subseteq U$. This completes the proof.
\end{proof}

We now turn our attention to a certain class of $T_0$ spaces for which $\mathcal{I}$-continuity is a necessary condition for the $\mc{I}$-convergence to be topological.

While any directed set $D$ with existing supremum automatically defines a net that $\mathcal{S}$-converges to $\bigvee D$, given $E \in \mathcal{I}_{\vee}$, one cannot always guarantee the existence of a net whose terms are in $E$ (or $\da E$) and that $\mathcal{I}$-converges to $\bigvee E$.
Finally, we decide to look at spaces in which every irreducible set can be `mimicked' by some directed subset of its lower closure in some sense, and $C$-spaces provide us with some inspiration.  A $T_0$ space $(X,\tau)$ is a \emph{$C$-space} if for any $x \in X$ and $U \in \tau$ with $x \in U$, there exists a $y \in U$ such that $x \in \intr_\tau (\ua y)$.  It is well-known that the Scott space of any continuous domain is a $C$-space (see, e.g., \cite{erne91,erne05}). With regards to irreducible sets, $C$-spaces have a very pleasing property:
\begin{propC}[{\cite[Lemma 6]{erne05}}] \label{prop: motivation for niceness}
%[Lemma 7.4]{zhaoho15}
Let $X$ be a $C$-space and $E \in \mathcal{I}$.  Then there exists a directed subset $D$ of $\da E$ with the same set of upper bounds as $E$.  In particular, $\bigsup D = \bigsup E$, if either exists.
\end{propC}
So the responsibility of having a net that $\mathcal{I}$-converges to the supremum of an irreducible set can be passed on to some directed subset of it.  Singling out this property, we formulate the following definition:
\begin{defi}[$\mathcal{DI}$ space] \label{defi: nice spaces}
A $T_0$ space $X$ is said to be \emph{$\mathcal{DI}$} if for each $E \in \mathcal{I}_{\vee}$ there exists a directed subset $D$ of $\da E$ such that $\bigsup D = \bigsup E$.
\end{defi}

\begin{example}\label{ex:nice}
\begin{enumerate}
\item[{\rm (1)}] Every poset endowed with the Alexandroff topology is $\mathcal{DI}$.
\item[{\rm (2)}] Any Scott space of a continuous poset is a $C$-space (see \cite[Theorem 4]{erne05}), hence $\mathcal{DI}$ by Proposition \ref{prop: motivation for niceness}.
\item[{\rm (3)}] Any $T_1$ space is $\mathcal{DI}$ since the only subsets having a supremum are singletons.
\item[{\rm (4)}] The Johnstone space \cite{johnstone81} is a $\mathcal{DI}$ because of the following fact: if $E \in \mathcal{I}_{\vee}$ and
\begin{enumerate}
\item $\bigvee E = (m,n)$ for some $m,n \in \mathbb{N}$, then $(m,n) \in E$.
\item $\bigvee E = (m, \infty)$ for some $m \in \mathbb{N}$, then $(m, \infty) \in E$ or $\da E \cap \{(m,n) \mid n \in \mathbb{N}\}$ is a directed set whose supremum is $(m, \infty)$.
\end{enumerate}
The Johnstone space is not a C-space due to the fact that the interior of any principal filter is an empty set.

\item[{\rm (5)}] The space $\Sigma \mathbb{J}^*$ given in Appendix~\ref{sec: appendix} is not $\mathcal{DI}$ since $\mathbb{J} \in \mathcal{I}_{\vee}$ but there is no directed subset of $\mathbb{J}$ having $\top = \bigvee \mathbb{J}$ as supremum.

\item[{\rm (6)}] $\Sigma\mathbb{D}$ given in Appendix~\ref{sec: appendix} is not $\mathcal{DI}$ since the lower set $\mathbb{N} \times \mathbb{N} \in \mathcal{I}_{\vee}$ contains no directed subset whose supremum is $\top = \bigvee (\mathbb{N} \times \mathbb{N})$.
\end{enumerate}
\end{example}

Although $\mathcal{DI}$ spaces seem to have circumvented the problem of defining certain nets out of irreducible sets with existing supremum, such spaces do not give any essentially new results as the remark below reveals.

\begin{rem}\label{rem:DI}
The following conditions hold for any $\mathcal{DI}$ space $(X,\tau)$:
\begin{enumerate}
\item $x \lli y$ if and only if $x \ll y$ in the poset $(X,\leq_\tau)$,
\item $n \siconv x$ in $(X,\tau)$ if and only if $n \xrightarrow{\mc{S}} x$ in $(X,\leq_\tau)$,
\item $(X,\tau)$ is $\mathcal{I}$-continuous if and only if $(X,\leq_\tau)$ is continuous.
\end{enumerate}
\end{rem}

All in all, we have arrived at:
\begin{thm}[Topological Scott Convergence Theorem]~
\label{thm: iter_lim_implies_SIcts}
\begin{enumerate}
\item[{\rm (1)}]
If $(X,\tau)$ is an $\mathcal{I}$-stable $\mathcal{I}$-continuous space, then the $\mathcal{I}$-convergence in $(X,\tau)$ is topological. In particular, $\tau_\mathcal{I}$ coincides with the topology generated by $\sigma_{\mathcal{I}}$.
\item[{\rm (2)}]
If $X$ is a $\mathcal{DI}$ space in which $\mathcal{I}$-convergence is topological, then $X$ is $\mathcal{I}$-continuous.
\end{enumerate}
\end{thm}

\section{Conclusion}
\label{sec: conclusion}
In this paper, we have proven a topological variant of the Scott Convergence Theorem as promised in the introduction.  The key strategy used in our approach is the $\mathcal{ID}$ replacement principle, i.e., replace the directed sets by irreducible sets for a given domain-theoretic definition.  Irreducible sets play an important role in topology, and particularly in domain theory.  Recently, (Scott) irreducible sets were actively employed to solve the so-called Ho-Zhao problem~\cite{hojungxi16}.
By examining what domain-theoretic results can be lifted to the higher plains of $T_0$ spaces, we manage to open up, otherwise uncharted, research areas in the study of $T_0$ spaces, e.g., new concepts like $\mathcal{I}$-stable spaces and the $\mathcal{DI}$ spaces.  We end our paper with some research problems, which we hope to spur research efforts in developing Lawson-style domain theory.

\paragraph{\textbf{Hunting for counterexamples.}} Theorem~\ref{thm: iter_lim_implies_SIcts}(1) relies on the interpolation property enjoyed by $\lli$ in an  $\mathcal{I}$-stable $\mathcal{I}$-continuous space.  It is natural to ask whether $\lli$ always satisfies the interpolation property for every $\mathcal{I}$-continuous space. If the answer is positive, then the $\mathcal{I}$-stability in Theorem \ref{thm: iter_lim_implies_SIcts}(1) can be removed. If otherwise, then another question arises, asking whether there is an $\mathcal{I}$-continuous space in which $\mathcal{I}$-convergence is not topological. The first place to look for possible examples of such $\mathcal{I}$-continuous spaces should be among the non-$\mathcal{I}$-stable ones.  However, at the moment, even an example of an $\mathcal{I}$-continuous space which is not $\mathcal{I}$-stable evades us. Our lack of ability in finding such an example leads us to the reality that Theorem \ref{thm: fund thm on I-stable spaces} is not entirely satisfactory.

\paragraph{\textbf{$\mathcal{I}$-stable, $\mathcal{I}$-continuous, sup-sober and $\mathcal{DI}$.}} In this paper, we have introduced different kinds of $T_0$ spaces, namely, $\mathcal{I}$-stable, $\mathcal{I}$-continuous and $\mathcal{DI}$, and involved a new variant of sober spaces called sup-sober.  Table~\ref{table1} (see Appendix~\ref{sec: appendix}) whose column headers are these four properties (listed in the following order: $\mathcal{I}$-stable, $\mathcal{I}$-continuous, sup-sober, $\mathcal{DI}$) collates as many examples of $T_0$ spaces currently known to us that possess a range of different  parity-configurations. Here, $+$ (respectively, $-$) indicates the presence (respectively, absence) of the corresponding property in that position.
The paucity of examples and counterexamples that distinguish these four topological properties points towards our patchy understanding of these properties and their interrelationships.  One such known interrelationship is exemplified by Theorem~\ref{thm: fund thm on I-stable spaces} ((1) $\implies$ (2)) which asserts that  ``$\mathcal{I}$-continuous and sup-sober imply $\mathcal{I}$-stable".  Because of this, no space of configurations ``$-~ +~ +~ +$'' or ``$- ~+ ~+ ~-$'' can exist.  Our ignorance is fully exposed by the lack of examples with the following configurations:
\begin{multicols}{4}
\begin{enumerate}[(1)]
\item ``$+ ~+~ +~ -$''
\item ``$+ ~+~ -~ -$''
\item ``$- ~+ ~- ~+$''
\item ``$-~ +~ -~ -$''
\item ``$-~ -~ +~ -$''
\item ``$- ~- ~- ~+$''
\item ``$-~ -~ -~ -$''
\end{enumerate}
\end{multicols}

\paragraph{\textbf{A closer look at $\mathcal{I}$-stable $\mathcal{I}$-continuous spaces.}} In view of the several pleasant properties that $\mathcal{I}$-stable $\mathcal{I}$-continuous spaces possess (e.g., the all-important interpolation property), it seems to us that such spaces deserve a closer look.  We ask the following questions:
\begin{enumerate}
\item What closure properties does the class of all $\mathcal{I}$-stable $\mathcal{I}$-continuous spaces possess?

\item For any domain $P$, the collection of sets of the form $\dua x := \{p \in P \mid x \ll p\}$ (where $x \in P$) form a base for the Scott topology on $P$~\cite[Proposition II-1.10]{gierzetal03}.
    Are there such convenient bases for $\mathcal{I}$-stable $\mathcal{I}$-continuous spaces?

\item The lattice of Scott opens of a dcpo is completely distributive if and only if it is a domain~\cite{Hoffmann1981}. Is there a similar characterisation for the lattice of $\tau_{\SI}$-opens for a space $(X,\tau)$ which has all irreducible suprema?
It is also interesting to ask if $\mathcal{I}$-continuity is a necessary and sufficient condition for the Scott derivative of an $\mathcal{I}$-stable space to be a $C$-space.

%\item The injective $T_0$ spaces are exactly the Scott spaces of continuous lattices~\cite{scott72a}.  Are the injective $\mathcal{I}$-stable spaces exactly the Scott derivative of $\mathcal{I}$-continuous spaces?
\end{enumerate}

\paragraph{\textbf{How far can we go in this research program?}}
%While the experience of developing the core of domain theory in the topological setting brings in new topological discoveries and refreshing perspectives from time to time, we are constantly forced to perform a reality-check -- are we doing things that differently from domain theory?  After all, as J. Goubault-Larrecq sums it up in the introduction of his book~\cite{goubaultlarrecq2013} that ``in several aspects, domain theory is \emph{topology done right}".
While some of the research findings that we presented here persuade us that domain theory motivates the creation of new topological concepts (e.g., $\mathcal{I}$-stable spaces), there is also confounding evidence that some of these `new' considerations (e.g., $\mathcal{DI}$ spaces) eventually lead us back to the good old domain theory!

We pause for a quick reflection.  Our current research methodology relies solely on the $\mathcal{ID}$ replacement principle, trusting that the collection of all irreducible sets is a salient substitute for the collection of all directed sets in the $T_0$ space setting.  Perhaps, this view is too narrow.  In view of this, we should aim to build a theory flexible enough to allow us to consider other kinds of collections of subsets of a $T_0$ space. There is evidence that the $\mathcal{M}$-subset system and $\mathcal{Z}$-subset selection are powerful tools of generalising domain theory~\cite{bandelterne83, bandelterne84, baranga96, erne81a, erne99, menon96, venugopalan86, yao11, zhao92, zhouzhao07}. One possible approach is to lift $\mathcal{M}$-subset system theory to the topological level!  Already in this paper we see some semblance of this for the specific case of $\mathcal{M} :=\mathcal{I}$, where the collection of all irreducible subsets of a $T_0$ space $X$ has been topologised with the lower Vietoris topology.  With some hard work, we should be able to realise the manifestation of $\mathcal{M}$- and $\mathcal{Z}$-theory in the topological realm.

\section*{Acknowledgement}

We are grateful to the anonymous referees for their valuable suggestions, which helped to eliminate several errors in former versions of this paper and to improve the presentation and results considerably.  We credit the first reviewer for his formulation of the statements and results as presented now in Section \ref{sec: in the light of Ztheory}, and in particular, Theorem~\ref{thm: fund thm on I-stable spaces}.

\newpage

\appendix \section{Some examples}
\label{sec: appendix}
\paragraph{\textbf{Notations.}} For a poset $P=(X,\leq)$,
\begin{itemize}
\item[] $\textup{A} P:= (X,\alpha(P))$ is the Alexandroff space of $P$,
\item[] $\Sigma P := (X,\sigma(P))$ is the Scott space of $P$,

\item[] $\Upsilon P := (X,\nu(P)$ is the (weak) upper space of $P$.
\end{itemize}
We define the following posets:
\begin{itemize}
\item $\mathbb{R}_2:=(\{1,2\}\times \mathbb{R})\cup \{\top\}$ ordered by $\leq$ defined as follows:
\begin{enumerate}[(i)]
\item $(m,n)\leq \top$ for each $(m,n)$ and
\item $(m,n)\leq (p,q)\t{ if and only if }(m=p \t{ and }n\leq q)$.
\end{enumerate}
\item $T=\{(1-\frac{1}{n},0)\mid n\in \mathbb{N}\}\cup \{(0,1),(1,1)\}$ ordered componentwise by $\leq$.

\item $\mathbb{D}=\brackets{\mathbb{N}\times (\mathbb{N}\cup \{\infty\})}\cup \{\top\}$ ordered by $\leq$ defined as follows:
\begin{enumerate}[(i)]
\item $(m,n)\leq \top$ for each $(m,n)$ and
\item $(m,n)\leq (p,q)\t{ if and only if }(m=p \t{ and }n\leq q) \t{ or }(m\leq p \t{ and }n\leq q=\infty)$.

\end{enumerate}

\item $\mathbb{J}$ is the Johnstone poset \cite{johnstone81}, and $\mathbb{J}^*$ is obtained by adding a top element.
\end{itemize}
We then have the following table containing classification of some examples.

\begin{table}[ht]
\centering
\caption{Some Examples}
\label{table1}
\def\arraystretch{1.4}
\begin{tabular}{|l|c|c|c|c|l|N}
\hline
$T_0$ space                     & $\mathcal{I}$-stable & $\mathcal{I}$-cont. & sup-sober         & $~~\mathcal{DI}~~$ & \multicolumn{1}{c|}{remarks}                                                                                      & \\ \hline
$T_1$ space $(X,\tau)$          & $+$                  & $+$                 & $+$               & $+$            & $\nu(P)\subseteq \tau\subseteq \sigma(P)=\alpha(P)$                                  & \\ \hline
$\Sigma P$ ($P$ is cont.)     & $+$                  & $+$                 & $+$               & $+$            & $\nu(P)\subseteq \sigma(P)\subseteq \alpha(P)$                                                                     & \\ \hline
$\textup{A} P$                 & $+$                  & a               & b & $+$            & \begin{tabular}[c]{@{}l@{}}$\nu(P)\subseteq \sigma(P)\subseteq \alpha(P)$\\ a. $A P$ is $\mathcal{I}$-cont. if and\\ \phantom{a. }only if $P$ is cont.\\ b. If $AP$ is sup-sober then\\ \phantom{b. }$AP$ is $\mathcal{I}$-cont.\end{tabular} & \\ \hline
$\textup{A} \mathbb{R}$        & $+$                  & $+$                 & $-$               & $+$            & $\nu(\mathbb{R})= \sigma(\mathbb{R})\subset \alpha(\mathbb{R})$                                                                     & \\ \hline
$\Sigma \mathbb{R}_2$        & $+$                  & $-$                 & $+$               & $+$            & $\nu(\mathbb{R}_2)=\sigma(\mathbb{R}_2)\subset \alpha(\mathbb{R}_2)$                                                                     & \\ \hline
$\textup{A} T$                 & $+$                  & $-$                 & $-$               & $+$            & $\nu(T)= \sigma(T)\subset \alpha(T)$                                                                     & \\ \hline
$\Sigma \mathbb{D}$             & $+$                  & $-$                 & $+$               & $-$            & $\nu(\mathbb{D})\subset \sigma(\mathbb{D})\subset \alpha(\mathbb{D})$                                                                     & \\ \hline
$\Sigma \mathbb{J}$             & $-$                  & $-$                 & $+$               & $+$            & $\nu(\mathbb{J})\subset \sigma(\mathbb{J})\subset \alpha(\mathbb{J})$                                        & \\ \hline
$\Sigma \mathbb{J}^*$           & $+$                  & $-$                 & $-$               & $-$            & $\nu(\mathbb{J}^*)\subset \sigma(\mathbb{J}^*)\subset \alpha(\mathbb{J}^*)$                                        & \\ \hline
\end{tabular}
\end{table}
\end{document}